\documentclass[11pt,a4paper]{article}

\usepackage[margin=1in]{geometry}

\title{Improved Strong Spatial Mixing for Colorings on Trees}

\author{Charilaos Efthymiou\thanks{Department of Computer Science, University of Warwick, UK. Supported by the Centre of Discrete Mathematics and its Applications (DIMAP), University of Warwick, EPSRC award EP/D063191/1.}
\and 
Andreas Galanis\thanks{Department of Computer Science, University of Oxford, UK. The research leading to these results has received funding from the European Research Council under the European Union's Seventh Framework Programme (FP7/2007-2013) ERC grant agreement no.\ 334828. The paper reflects only the authors' views and not the views of the ERC or the European Commission. The European Union is not liable for any use that may be made of the information contained therein.}
\and
Thomas P. Hayes\thanks{Department of Computer Science, University of New Mexico, USA. Partially supported
by NSF CAREER award CCF-1150281.}
\and
Daniel \v{S}tefankovi\v{c}\thanks{Department of Computer Science, University of Rochester, USA. Research supported in part by NSF grant CCF-1563757.}
\and
Eric Vigoda\thanks{School of Computer Science, Georgia Institute of Technology, USA. Research supported in part by NSF grants CCF-1617306 and CCF-1563838.}
}

\usepackage{amsmath,amssymb,amsthm}
\usepackage{hyperref}
\hypersetup{colorlinks=true,linkcolor=blue, citecolor=blue}
\usepackage{filecontents}

\newtheorem{theorem}{Theorem}
\newtheorem{lemma}[theorem]{Lemma}
\newtheorem{observation}[theorem]{Observation}

\newtheorem{definition}[theorem]{Definition}

\def\xb{\mathbf{x}}

\def\zb{\mathbf{z}}
\def\Mb{\mathbf{M}}
\def\ub{\mathbf{u}}
\def\T{\ensuremath{\intercal}}
\def\pib{\boldsymbol{\pi}}
\def\gammab{\boldsymbol{\gamma}}

\newcommand{\integers}{\mathbb{Z}}

\newcommand{\fptas}{\mathsf{FPTAS}}

\newcommand{\Treed}{\mathbb{T}_{d}}
\def\dist{\mathrm{dist}}

\newcommand{\wsm}{\mathsf{WSM}}
\newcommand{\ssm}{\mathsf{SSM}}
\newcommand{\WSM}{\mathsf{WSM}}
\newcommand{\SSM}{\mathsf{SSM}}
\newcommand{\poly}{\mbox{poly}}

\begin{document}

\maketitle

\begin{abstract}
Strong spatial mixing (SSM) is a form of correlation decay that has played an essential role in the design of approximate counting algorithms for spin systems.  A notable example is the algorithm of Weitz (2006) for the hard-core model on weighted independent sets. We study SSM for the $q$-colorings problem on the infinite ($d$+1)-regular tree. Weak spatial mixing (WSM) captures whether the influence of the leaves on the root vanishes as the height of the tree grows.  Jonasson (2002) established WSM when $q>d+1$.  In contrast, in SSM, we first fix a coloring on a subset of internal vertices, and we again ask if the influence of the leaves on the root is vanishing.  It was known that SSM holds on the $(d+1)$-regular tree when $q>\alpha d$ where $\alpha\approx 1.763...$ is a constant that has arisen in a variety of results concerning random colorings. Here we improve on this bound by showing SSM for $q>1.59d$. Our proof establishes an $L^2$ contraction for the BP operator. For the contraction we bound the norm of the BP Jacobian by exploiting  combinatorial properties of the coloring of the tree.
\end{abstract}

\section{Introduction}

Consider random $q$-colorings of the complete tree $T_h$ of height $h$ with branching factor $d$.  
Does the influence of the leaves on the root decay to zero in the limit as the height grows?
If so, this corresponds to weak spatial mixing, which we will define more precisely momentarily.

Now suppose we fix the coloring $\tau$ for a subset of internal vertices. Is it still the case that the influence of the leaves on the root decay to zero as the height grows?   
One might intuitively expect that these internal ``agreements'' defined by $\tau$ only help in the
sense that the influence of the leaves decrease, however
this problem is much more challenging; it corresponds to strong spatial mixing, which is the focus of
this paper.

For statistical physics models, the key algorithmic problems are 
the counting problem of estimating the partition function and 
the problem of sampling from the Gibbs distribution, which corresponds to the
equilibrium state of the system.
Strong spatial mixing ($\ssm$) is a key property of the system for the design of efficient counting/sampling algorithms.

$\SSM$ has a variety of algorithmic implications.  A direct consequence of $\SSM$ on amenable graphs,
such as the integer lattice $\integers^d$, is fast mixing of the Glauber dynamics, which is the
simple Markov chain that updates the spin at a randomly chosen vertex in each step, see, e.g.~\cite{MO1,MO2,Cesi,DSVW,GMP,BS,BCSV}.
$\SSM$ also plays a critical role in the efficiency of correlation-decay techniques of Weitz~\cite{Weitz}
which yields an $\fptas$ for the partition function of the hard-core model in the tree uniqueness region;
this approach has been extended to 2-spin antiferromagnetic models~\cite{LLY} and other 
interesting examples, e.g.,~\cite{LL}; note, the approach of Barvinok~\cite{Barvinok}
utilizing a zero-free region of the partition function in the complex plane has recently been 
extended to the same range of parameters for the hard-core model~\cite{PaR, PR}.

The fundamental question in statistical physics is the uniqueness/non-uniqueness phase transition
which corresponds to whether long-range correlations persist or die off, in the limit as the volume of the system
tends to infinity.    In the uniqueness region the correlations die off, which corresponds to {\em weak spatial mixing} ($\WSM$).
While $\WSM$ (or equivalently uniqueness) is a notoriously challenging problem on 
the 2-dimensional integer lattice $\integers^2$ (e.g., see the recent breakthrough
work of Beffara and Duminil-Copin~\cite{BDC} for the ferromagnetic Potts model), the corresponding $\wsm$ problem
on the infinite $(d+1)$-regular tree $\Treed$, known as the Bethe lattice, is typically simpler since it can be analyzed using recursions due to the absence of cycles  (e.g., see Kelly~\cite{Kelly}
for the hard-core model).   However, for the colorings problem, which is the focus of this paper, even $\wsm$ is far from trivial
on the regular tree~\cite{Jonasson}. In fact, for the closely related antiferromagnetic Potts model the precise range of parameters for $\wsm$ is only known for fixed values of $q,d$~\cite{GGY}.

The focus of this paper is on these correlation decay properties on the infinite $(d+1)$-regular tree~$\Treed$
for the {\em colorings} problem.  We give an informal definition of $\wsm$ and $\ssm$, and refer the 
interested reader to Section~\ref{sec:definitions} for formal definitions.

Let $T_h$ denote the complete tree of height $h$ where all internal vertices have degree $d+1$.
For integer $q\geq 3$, let $\mu_h$ denote the uniform distribution over proper (vertex) $q$-colorings of $T_h$.
Consider a pair of sequences of colorings $(\eta_h)$ and $(\eta'_h)$ for the leaves of $T_h$. 
Let $p_h$ and $p'_h$ denote the marginal probability that the root receives 
a specific color $c$ under $\mu_h$ conditional on the leaves having the fixed coloring $\eta_h$ and $\eta'_h$, respectively.
Roughly, if $\lim_{h\rightarrow\infty} |p_h-p'_h|= 0$ for all sequences $(\eta_h), (\eta'_h)$ and colors $c$,  then we say $\wsm$ holds (see also Section~\ref{sec:definitions}).
Jonasson~\cite{Jonasson} proved that $\wsm$ holds when $q\geq d+2$.  
When $q\leq d+1$, the pair of boundary conditions can actually ``freeze'' the color at the root;  moreover, Brightwell and Winkler~\cite{BW} showed that there are multiple semi-translation invariant Gibbs measures on $\Treed$
when $q\leq d$.  

Now consider an arbitrary coloring $\tau$ for a subset $S\subset\Treed$.  Let $r_h$ 
and $r'_h$ denote the marginal probability that the root receives 
color $c$ under $\mu_h$ conditional on $\eta_h\cup\tau$ and $\eta'_h\cup\tau$, respectively.
If these limits are the same then we say $\ssm$ holds.  The challenge of establishing $\ssm$ is illustrated
by the fact that if $\wsm$ holds then we know that $\lim_{h\rightarrow\infty} p_h = 1/q$ but that
is not necessarily the case in the $\SSM$ setting.

Ge and \v{S}tefankovi\v{c}~\cite{GS} proved that $\ssm$ holds on $\Treed$ when $q>\alpha d$ where $\alpha\approx 1.763...$ is the root of $\frac{1}{\alpha}\exp(1/\alpha)=1$. Gamarnik, Katz, and Misra~\cite{GKM} extended this result to arbitrary triangle-free graphs of maximum degree $d$, under the same condition on $q$. Recent work of Liu, Sinclair, and Srivistava \cite{LSS} builds upon \cite{GKM} together with the approximate counting approach of \cite{Barvinok,PaR} to obtain an $\fptas$ for counting colorings of triangle-free graphs when $q>\alpha d$.
Prior to these works, Goldberg, Martin, and Paterson~\cite{GMP} established the above form of 
$\SSM$ on triangle-free amenable\footnote{Roughly, a graph is amenable if
for every subset $S$ of vertices, the neighborhood satisfies $|N(S)|\leq\poly(|S|)$} graphs, also when $q>\alpha d$.   
In addition to the above results, the threshold $\alpha\approx 1.76\hdots$ has arisen in numerous rapid mixing results, e.g.,~\cite{DF,Hayes,DFHV}.

Our main result presents the first substantial improvement on the $1.76...$ threshold of \cite{GS}, we establish  $\SSM$ on the tree when  $q>1.59 d$. We state a somewhat informal version of our main theorem here, the formal version will be given once we define more precisely $\SSM$, cf. Theorem~\ref{thm:main} below. 
\begin{theorem}[Informal version of Theorem~\ref{thm:main}]\label{thm:maininformal}
There exists an absolute constant $\beta>0$ such that, for all positive integers $q,d$ satisfying $q\geq 1.59d+\beta$, the $q$-coloring model exhibits strong spatial mixing on the regular tree $\mathbb{T}_d$. 
\end{theorem}
We remark that the constant $1.59$ in Theorem~\ref{thm:maininformal} can be replaced with any $\alpha'>1$ satisfying \[\frac{1}{\alpha'}\exp\Big(\frac{1}{\alpha'}\Big)\exp\Big(-\frac{1}{\alpha' - 1 +\exp\big(\frac{1}{\alpha'-1}\big)}\Big)<1,\] 
the smallest such value up to four decimal digits is $1.5897$.

We give an overview of our proof approach in Section~\ref{sec:proofapproach} after formally defining $\SSM$ in
Section~\ref{sec:definitions} and stating the formal version of Theorem~\ref{thm:maininformal}.  We then present detailed proofs of the three main lemmas in Section~\ref{sec:proofapproach}.

\section{Definitions}
\label{sec:definitions}

Let $q\geq 3$ be an integer and $G=(V,E)$ be a graph. A proper $q$-coloring of $G$ is an assignment $\sigma:V\rightarrow [q]$ such that for every $(u,v)\in E$ it holds that $\sigma(u)\neq \sigma(v)$.  We use $\Omega_G$ to denote  the set of all proper $q$-colorings of $G$  and $\mu_G$ to denote the uniform probability distribution on $\Omega_G$ (provided that $\Omega_G$ is non-empty). 

For $\sigma\in \Omega_G$ and a set $\Lambda\subset V$, we use $\sigma_\Lambda$ to denote the restriction of $\sigma$ to $\Lambda$. When $\Lambda$ consists of a single vertex $v$, we will often use the shorthand $\sigma_v$ to denote the color of $v$ under $\sigma$. We say that an assignment $\eta:\Lambda \rightarrow [q]$ is \emph{extendible} if there exists a coloring $\sigma\in \Omega_G$ such that $\sigma_\Lambda=\eta$. 

We can now formally define $\SSM$.
\begin{definition}\label{def:SSM}
Let $\zeta: \mathbb{Z}_{\geq 0}\rightarrow [0,1]$ be a real-valued function on the positive integers.  

The $q$-coloring model exhibits \emph{strong spatial mixing}, denoted $\ssm$, on a finite graph $G=(V,E)$ with decay rate $\zeta(\cdot)$ iff for every $v\in V$, for every $\Lambda\subset V$, for any two extendible assignments $\eta,\eta':\Lambda \rightarrow [q]$ and any color $c\in [q]$ it holds that
\begin{equation}
\label{eq:SSM}
\big| \mu_G(\sigma_v=c \mid \sigma_\Lambda=\eta)-\mu_G(\sigma_v=c \mid \sigma_\Lambda=\eta')\big|\leq \zeta\big(\mathrm{dist}(v,\Delta)\big),
\end{equation}
where $\Delta\subseteq \Lambda$ denotes the set of vertices where $\eta$ and $\eta'$ disagree.

In the case where $G$ is infinite, we say that the $q$-coloring model exhibits strong spatial mixing on $G$ with decay rate $\zeta(\cdot)$ if it exhibits strong spatial mixing on every finite subgraph of $G$ with decay rate $\zeta(\cdot)$.
\end{definition}

The definition of weak spatial mixing has one modification: in the RHS of \eqref{eq:SSM} we replace $\dist(v,\Delta)$ by the weaker condition $\dist(v,\Lambda)$.  $\wsm$ says that the influence of a pair of boundary conditions decays at rate $\zeta(\cdot)$ in the
distance to the boundary $\Lambda$.  In $\SSM$ the pair of boundaries $\eta,\eta'$ might only differ on a subset $\Delta\subset\Lambda$; do these fixed ``agreements'' on $\Lambda\setminus\Delta$ influence the marginal at $v$?
If $\SSM$ holds then the difference in the marginal at $v$ decays at rate $\zeta(\cdot)$ in the distance to the 
``disagreements'' in $\eta,\eta'$.

With these definitions in place, we are now ready to give the formal version of Theorem~\ref{thm:maininformal}.
\begin{theorem}\label{thm:main}
There exists an absolute constant $\beta>0$ such that, for all positive integers $q,d$ satisfying $q\geq 1.59d+\beta$, the $q$-coloring model exhibits strong spatial mixing on the regular tree $\mathbb{T}_d$ with exponentially decaying rate. 

That is, there exist constants $\alpha,C>0$ and  a function $\zeta$ satisfying $\zeta(\ell)\leq C \exp(-\alpha \ell)$ for all integers $\ell\geq 0$ such that for all finite subtrees $T$ of $\mathbb{T}_d$ the $q$-coloring model exhibits strong spatial mixing on $T$ with decay rate $\zeta$.
\end{theorem}

\section{Proof Approach}\label{sec:proofapproach}

For a set $\Lambda\subset V$ and an extendible assignment $\eta:\Lambda\rightarrow [q]$, we use $\boldsymbol{\pi}_{G, v,\eta}$ to denote the $q$-dimensional probability vector whose entries give the marginal distribution of colors at $v$ under the boundary condition $\eta$, i.e., for a color $c\in [q]$, the $c$-th entry of  $\pib_{G,v,\eta}$ is given by $\mu_G(\sigma_v=c\mid \sigma_\Lambda=\eta)$. 

The key ingredient to prove Theorem~\ref{thm:main} is the following. 
\begin{theorem}\label{thm:main2}
There exist absolute constants  $\beta>0$ and $U\in(0,1)$ such that the following holds for all positive integers $q,d$ satisfying $q\geq 1.59d+\beta$.  

Let $T=\hat{\mathbb{T}}_{d,h, \rho}$ be the $d$-ary tree with height $h$ rooted at $\rho$, $\Lambda$ be a subset of the vertices of $T$, and $\eta,\eta':\Lambda\rightarrow [q]$ be two extendible assignments of $T$ with $\dist(\rho,\Delta)\geq 3$ where $\Delta\subseteq \Lambda$ is the set of vertices where $\eta$ and $\eta'$ disagree. Let $v_1,\hdots, v_{d}$ be the children of $\rho$ and for $i\in [d]$ let $T_i=(V_i,E_i)$ be the subtree of $T$ rooted at $v_i$ which consists of all descendants of $v_i$ in $T$. Then
\begin{equation*}
\big\| \pib-\pib'\big\|_2^2\leq U\max_{i\in [d]}\big\| \pib_{i}-\pib_{i'}\big\|_2^2,
\end{equation*} 
where  $\pib=\pib_{T,\rho, \eta}$, $\pib'=\pib_{T, \rho, \eta'}$ and for $i\in [d]$ we denote $\pib_i=\pib_{T_i,v_i, \eta(\Lambda\cap V_i)}$, $\pib_i'=\pib_{T_i, v_i, \eta'(\Lambda\cap V_i)}$.
\end{theorem}

Intuitively, Theorem~\ref{thm:main2} says that  disagreements between $\eta$ and $\eta'$ have smaller impact on the marginals as we move upwards on the tree. More precisely, the marginals of the root under $\eta$ and under $\eta'$ are closer in $L^2$ distance than the distance between the marginals of any child (under the induced distributions on the subtrees hanging from them). 

Using Theorem~\ref{thm:main2}, the proof of Theorem~\ref{thm:main} of strong spatial mixing follows from rather standard considerations, the proof can be found in Section~\ref{sec:proofofmaintheorem}. In the following section, we focus on the more interesting proof of  Theorem~\ref{thm:main2} and explain the new aspects of our analysis.

\subsection{The three main lemmas}
In this section, we lay down the main technical steps in proving Theorem~\ref{thm:main2}. In particular, we will assume throughout that, for appropriate integers $q,d,h$, $T=\hat{\mathbb{T}}_{d,h, \rho}$ is the $d$-ary tree with height $h$ rooted at $\rho$, $\Lambda$ is a subset of the vertices of $T$, and $\eta,\eta':\Lambda\rightarrow [q]$ are two extendible assignments of $T$ with $\mathrm{dist}(\rho,\Delta)\geq 3$ where $\Delta\subseteq \Lambda$ is the set of vertices where $\eta$ and $\eta'$ disagree.   We als let $v_1,\hdots, v_{d}$ be the children of $\rho$ and for $i\in [d]$ let $T_i=(V_i,E_i)$ be the subtree of $T$ rooted at $v_i$ which consists of all descendants of $v_i$ in $T$.

To prove Theorem~\ref{thm:main2}, we will use tree recursions to express the marginal at the root in terms of the marginals at the children (as in previous works on WSM/SSM, see, e.g., \cite{BW,GS,GGY}).  This recursion is the well-known {\em Belief Propagation} (BP) equation \cite{Pearl}; our proof of  Theorem \ref{thm:main2} will be based on bounding appropriately the gradient of the BP equations. 
The  new ingredient in our analysis is that we  incorporate the combinatorial structure of agreements close to the root into a refined $L^2$ 
analysis of the  gradient.

Prior to delving into the analysis, we first describe the BP equation for the colorings model. Following the notation of Theorem~\ref{thm:main2}, let $\pib=\pib_{T,\rho, \eta}$, $\pib'=\pib_{T, \rho, \eta'}$ be the marginal distributions at the root of the tree $T$ under the boundary conditions $\eta$ and $\eta'$, respectively. Similarly,  for $i\in [d]$, let $\pib_i$, $\pib_i'$ be the marginals at the root $v_i$ of the subtree $T_i$ under $\eta(\Lambda\cap V_i)$ and  $\eta'(\Lambda\cap V_i)$, respectively.

 We can  now relate the distribution $\pib$ with the distributions $\{\pib_i\}_{i\in [d]}$ (and similarly, $\pib'$ with the distributions $\{\pib_i'\}_{i\in [d]}$) as follows. For $q$-dimensional probability vectors  $\xb_1,\hdots, \xb_d$ and a color $c\in [q]$, let $f_c$ be the function
\begin{equation}\label{def:fc}
f_c(\xb_1,\hdots, \xb_d)=\frac{\prod_{i\in [d]}\big(1-x_{i,c}\big)}{\sum_{j\in [q]}\prod_{i\in [d]}\big(1-x_{i,j}\big)},
\end{equation}
where, for $i\in [d]$ and $j\in [q]$, $x_{i,j}$ denotes the $j$-th entry of the vector $\xb_{i}$.  Then, with $\pi_c$ and $\pi_c'$ denoting the $c$-th entries of $\pib$ and $\pib'$, we have that 
\begin{equation}\label{eq:tree1recursion}
\begin{aligned}
\pi_c&=\mu_T(\sigma_\rho=c\mid \sigma_\Lambda=\eta)=f_c(\pib_{1},\hdots,\pib_{d}),\\
\pi_c'&=\mu_T(\sigma_\rho=c\mid \sigma_\Lambda=\eta')=f_c(\pib_{1}',\hdots,\pib_{d}').
\end{aligned}
\end{equation}
The functions $\{f_c\}_{c\in [q]}$ correspond to the BP equations for the coloring model.

We are now ready to describe in more detail our SSM analysis. Specifically, to get a bound on the norm $\left\|\pib-\pib'\right\|_2$, we will study the gradient of $f_c$  as we change the arguments $(\pib_{1},\hdots,\pib_{d})$ to $(\pib_{1}',\hdots,\pib_{d}')$ along the line connecting them. Our gradient analysis will take account of the following combinatorial notions. 
\begin{definition}
A vertex $v$ of $T$ is called \emph{frozen} under $\eta$ if $v\in \Lambda$ and \emph{non-frozen} otherwise. For a non-frozen vertex $v$ of $T$, a color $k$ is \emph{blocked} for $v$ (under $\eta$) if  there is a neighbor $u\in \Lambda$ of $v$ such that  $\eta(u)=k$; the color is called \emph{available} for $v$ otherwise. 
\end{definition}
\begin{observation}\label{obs:sameblocked}
In the setting of Theorem~\ref{thm:main2}, we have that the disagreements between $\eta$ and $\eta'$ occur at distance at least 3 from the root. It follows that the set of the root's children that are frozen as well as the set of blocked colors for each of the non-frozen children are identical under both $\eta$ and $\eta'$. 
\end{observation}
We will utilize that the gradient components that correspond to either frozen children or blocked colors can be disregarded since, by Observation~\ref{obs:sameblocked}, the corresponding arguments in \eqref{eq:tree1recursion} are fixed to the same value. Namely, we will  track, for each color $c$,  the fraction of non-frozen 
children which have color $c$ available. This will allow us in the upcoming Lemma~\ref{lem:mainlemma1} to aggregate accurately the gradient components  corresponding to color $c$. 
The following definitions setup some relevant notation.
\begin{definition}\label{def:gammas}
Let $D$ be the indices of the children of the root which are non-frozen under $\eta$ and $\eta'$. For a color $c\in [q]$, let  $\gamma_c\in[0,1]$ be the \emph{fraction} of indices $i\in D$ such that color $c$ is available for $v_i$ under $\eta$ and $\eta'$ (cf.  Observation~\ref{obs:sameblocked}). Let $\gammab$ and  $\sqrt{\gammab}$ be the $q$-dimensional vector with entries $\{\gamma_c\}_{c\in [q]}$ and $\{\sqrt{\gamma_c}\}_{c\in [q]}$, respectively.  
\end{definition}
Intuitively, if $\gamma_c$ is close to 0, color $c$ is blocked at a lot of the children  and hence the distance $\left\|\pib-\pib'\right\|_2$ at the root should not depend a lot on the color $c$ (since most components of the gradient corresponding to color $c$ are zero).

The following couple of definitions will be relevant for capturing more precisely the gradient of the functions $\{f_c\}_{c\in [q]}$. To begin with, the gradient will actually turn out to be related to the value of $f_c$ as we move along the line $(\pib_{1},\hdots,\pib_{d})$ to $(\pib_{1}',\hdots,\pib_{d}')$. More precisely, we have the following definition.
\begin{definition}\label{def:hatpi}
For $t\in[0,1]$, let $\hat{\pib}(t)=\{\hat{\pi}_c(t)\}_{c\in [q]}$ be the $q$-dimensional probability vector whose $c$-th entry is given by $f_c\big(t\pib_1+(1-t)\pib_1',\hdots, t\pib_d+(1-t)\pib_d'\big)$.
\end{definition}
Note that $\hat{\pib}(1)=\pib$ and $\hat{\pib}(0)=\pib'$; in this sense, we can think of the vector $\hat{\pib}(t)$ as having the marginals at the root as we interpolate between $(\pib_{1},\hdots,\pib_{d})$ to $(\pib_{1}',\hdots,\pib_{d}')$.

The next definition will  be relevant for bounding the $L^2$ norm of the gradient along the line connecting to $(\pib_{1},\hdots,\pib_{d})$ to $(\pib_{1}',\hdots,\pib_{d}')$. The bound will be in terms of the ``marginals'' at the root, as captured by the vector $\hat{\pib}(t)$ (cf. Definition~\ref{def:hatpi}), and the availability of the $q$ colors at the children, as captured by the vector $\gammab$ (cf. Definition~\ref{def:gammas}). In particular, we will be interested in the $L^2$ norm of the following matrix, which is an idealized version to the Jacobian of the BP equation (see \eqref{eq:partialderivatives} for the precise formula).\footnote{For a square matrix $\Mb$, we use $\left\|\Mb\right\|_2$ to denote its $L^2$ norm, i.e., $\left\|\Mb\right\|_2=\max_{\left\|\xb\right\|_2=1} \left\|\Mb\xb\right\|_2$. A fact that will be useful later is that $\left\|\Mb\right\|_2=\max_{\left\|\xb\right\|_2=1} \left\|\xb^{\T}\Mb\right\|_2$, even for non-symmetric matrices $\Mb$.}
\begin{definition}\label{def:gradientmatrix}
Let $\hat{\pib},\hat{\gammab}$ be $q$-dimensional vectors with non-negative entries. The matrix $\Mb_{\hat{\pib},\hat{\gammab}}$ corresponding to the vectors $\hat{\pib},\hat{\gammab}$ is given by $\big(\mathrm{diag}(\hat{\pib})-\hat{\pib} \hat{\pib}^{\T}\big)\mathrm{diag}(\hat{\gammab}\big)$.\footnote{For a vector $\mathbf{\xb}$, $\mathrm{diag}(\xb)$ denotes the diagonal matrix with the entries of $\xb$ on the diagonal.}
\end{definition}

Our first main lemma shows how to bound the distance between the marginals at the root under $\eta$ and $\eta'$, i.e., $\big\| \pib-\pib'\big\|_2^2$, in terms of the aggregate distance at the children.  The new ingredient in our bound is to account more carefully for the availability of the colors at the children (i.e., the vector $\gammab$).
\begin{lemma}\label{lem:mainlemma1}
Let $q,d$ be positive integers so that $q\geq d+2$.   Then
\begin{equation*}
\big\| \pib-\pib'\big\|_2^2\leq |D|K^2\sum_{i\in [d]}\big\| \pib_{i}-\pib_{i'}\big\|_2^2 \qquad \mbox{ where $K:=\frac{1}{1-\tfrac{1}{q-d}}\max_{t\in (0,1)} \left\|\Mb_{\hat{\pib}(t),\sqrt{\gammab}}\right\|_2$,}
\end{equation*} 
where $D,\gammab, \sqrt{\gammab}$ are as in Definition~\ref{def:gammas}, $\hat{\pib}(t)$ is as in Definition~\ref{def:hatpi}, and  $\Mb_{\hat{\pib}(t),\sqrt{\gammab}}$ is as in Definition~\ref{def:gradientmatrix}. 
\end{lemma}
Given Lemma~\ref{lem:mainlemma1}, we are left with obtaining a good upper bound on the norm $\big\|\Mb_{\hat{\pib}(t),\sqrt{\gammab}}\big\|_2$ that takes advantage of the presence of the vector $\gammab$. It is not hard to see that the $L^2$ norm of the  matrix $\big(\mathrm{diag}(\hat{\pib})-\hat{\pib} \hat{\pib}^{\T}\big)$ is bounded by $\max_{j\in [q]}\hat{\pi}_j$. The following result can be seen as a generalisation of this fact, which is however significantly more involved to prove. The proof is given in Section~\ref{sec:mainlemma2}.
\begin{lemma}\label{lem:mainlemma2}
Let $q$ be a positive integer, $\hat{\pib}$ be a $q$-dimensional probability vector and $\hat{\gammab}$ be a $q$-dimensional vector with non-negative entries which are all bounded by 1. Then,  the $L^2$ norm of the matrix $\Mb_{\hat{\pib},\hat{\gammab}}=\big(\mathrm{diag}(\hat{\pib})-\hat{\pib} \hat{\pib}^{\T}\big)\mathrm{diag}(\hat{\gammab}\big)$ satisfies
\[\left\|\Mb_{\hat{\pib},\hat{\gammab}}\right\|_2\leq \frac{1}{2}\max_{j\in [q]} \hat{\pi}_j\big(1+(\hat{\gamma}_j)^2\big),\]
where $\{\hat{\pi}_j\}_{j\in [q]}, \{\hat{\gamma}_j\}_{j\in [q]}$ are the entries of $\hat{\pib},\hat{\gammab}$, respectively.
\end{lemma}

The final component of our proof is to utilize the bound in Lemma~\ref{lem:mainlemma2} to  derive an  upper bound on the norm of the matrix $\Mb_{\hat{\pib}(t),\sqrt{\gammab}}$ appearing in Lemma~\ref{lem:mainlemma1}. To prove Theorem~\ref{thm:main2}, we  roughly need to show that the norm is bounded by $1/|D|$. We show that this is indeed the case in Section~\ref{sec:mainlemma3}.
\begin{lemma}\label{lem:mainlemma3}
There exist absolute constants  $\beta>0$ and $K'\in(0,1)$ such that the following holds for all positive integers $q,d$ satisfying $q\geq 1.59d+\beta$.   

Let $\gammab,\hat{\pib}(t)$ be the $q$-dimensional vectors of Definitions~\ref{def:gammas} and~\ref{def:hatpi}, respectively. Then, for all $t\in [0,1]$ and all colors $k\in [q]$, it holds that
\begin{equation*}
\frac{1}{2}\hat{\pi}_k(t)(1+\gamma_k)<K'/|D|,
\end{equation*}  
where $D$ is the set of non-frozen children of $\rho$ under $\eta$ and $\eta'$.
\end{lemma}
Assuming Lemmas~\ref{lem:mainlemma1},~\ref{lem:mainlemma2} and~\ref{lem:mainlemma3} for now, we next conclude the proof of Theorem~\ref{thm:main2}.
\begin{proof}[Proof of Theorem~\ref{thm:main2}]
Let $U':= (1+K')/2$ where $K'\in (0,1)$ is the constant in Lemma~\ref{lem:mainlemma3}. Let $\beta>0$ be a sufficiently large constant so that, for all $q\geq 1.59d+\beta$,   the conclusion of Lemma~\ref{lem:mainlemma3} applies  and $\frac{1}{1-\frac{1}{q-d)}}K'<U'$. We will show that
\begin{equation}\label{eq:5tg5gygyy44f4}
\big\| \pib-\pib'\big\|_2^2\leq U\max_{i\in [d]}\big\| \pib_{i}-\pib_{i'}\big\|_2^2, \mbox{ with } U:=(U')^2.
\end{equation}
Indeed, by Lemmas~\ref{lem:mainlemma1},~\ref{lem:mainlemma2} and~\ref{lem:mainlemma3}, we have that
\[\big\| \pib-\pib'\big\|_2^2\leq \frac{U}{|D|}\sum_{i\in[d]}\big\| \pib_{i}-\pib_{i}'\big\|_2^2.\]
Note that an index $i\notin D$ corresponds to a frozen child $v_i$ and therefore $\pib_{i}=\pib_{i}'$ for all $i\notin D$ and hence 
\[\frac{1}{|D|}\sum_{i\in[d]}\big\| \pib_{i}-\pib_{i}'\big\|_2^2\leq \max_{i\in [d]}\big\| \pib_{i}-\pib_{i'}\big\|_2^2,\]
proving \eqref{eq:5tg5gygyy44f4}. This completes the proof of Theorem~\ref{thm:main2}.
\end{proof}

\section{Bound on the matrix norm: proof of Lemma~\ref{lem:mainlemma2}}\label{sec:mainlemma2}
In this section, we prove Lemma~\ref{lem:mainlemma2}.
\begin{proof}[Proof of Lemma~\ref{lem:mainlemma2}]
For this proof, it will be convenient to simplify notation and use $\pib$ instead of $\hat{\pib}$ and $\gammab$ instead of $\hat{\gammab}$, so that $\Mb_{\pib,\gammab}$ becomes $\big(\mathrm{diag}(\pib)-\pib \pib^{\T}\big)\mathrm{diag}(\gammab\big)$. Let $C:=\frac{1}{2}\max_{j\in [q]} \pi_j(1+\gamma_j^2)$. We will establish  that $\left\|\Mb_{\pib,\gammab}\right\|_2\leq C$ by showing that for an arbitrary $q$-dimensional vector $\xb$ it holds that 
\begin{equation}\label{eq:main1main1}
\left\|\xb^{\T}\Mb_{\pib,\gammab}\right\|_2^2\leq C^2\left\|\xb\right\|_2^2.
\end{equation} 
We will focus on proving \eqref{eq:main1main1} in the case where the entries of the vector $\gammab$ are all nonnegative and strictly less than one; the case where some of the entries of $\gammab$ are equal to 1 follows from the continuity of \eqref{eq:main1main1} with respect to $\gammab$.

So, assume that $\gamma_j\in [0,1)$ for all $j\in [q]$. Observe that 
\[\left\|\xb^{\T}\Mb_{\pib,\gammab}\right\|^2_2=\sum_{j\in [q]}\pi_j^2\gamma_j^2(x_j-w)^2 \mbox{ where } w:=\sum_{j\in [q]}\pi_j x_j.\]
Let $y_j= x_j -w$ for $j\in[q]$. Since $\pib$ is a probability vector, we have 
\[\sum_{j\in [q]}\pi_j y_j=0.\] 
Moreover, we can rewrite \eqref{eq:main1main1} as
\begin{equation}\label{eq:y4byh65r}
\sum_{j\in [q]}\frac{\pi_j^2\gamma_j^2}{C^2}y^2_j\leq \sum_{j\in [q]}(y_j+w)^2.
\end{equation}
Note that the function $f(z)=\sum_{j\in [q]}(y_j+z)^2$ achieves its minimum for $z^*=-\frac{1}{q}\sum_{j\in[q]} y_j$ and $f(z^*)=\sum_{j\in [q]}y_j^2-\frac{1}{q}\big(\sum_{j\in [q]}y_j\big)^2$. Hence, to prove \eqref{eq:y4byh65r} (and therefore \eqref{eq:main1main1}), it suffices to show that
\begin{equation}\label{eq:4rf4rf4678888}
\bigg(\sum_{j\in [q]}y_j\bigg)^2\leq q \sum_{j\in [q]}\frac{y_j^2}{A_j}, \mbox{ where } A_j:=\frac{C^2}{C^2-\pi_j^2\gamma_j^2}
\end{equation}
Note that the $A_j$'s are well-defined and greater than 1 for all $j\in [q]$ by our assumption that $\gamma_j\in [0,1)$, cf. the argument below \eqref{eq:main1main1}. Using that $\sum_{j\in [q]}\pi_j y_j=0$, we therefore obtain that \eqref{eq:4rf4rf4678888} is equivalent to
\begin{equation}\label{eq:4rf4rf4678888b}
\bigg(\sum_{j\in [q]}y_j(1+t\pi_j)\bigg)^2\leq q \sum_{j\in [q]}\frac{y_j^2}{A_j}, \mbox{ where } A_j:=\frac{C^2}{C^2-\pi_j^2\gamma_j^2},
\end{equation}
for any real number $t$ --- we will specify $t$ soon (cf. the upcoming \eqref{eq:tchoice}). In particular, by the Cauchy-Schwarz inequality, we have
\[\bigg(\sum_{j\in [q]}y_j(1+t\pi_j)\bigg)^2\leq \sum_{j\in [q]}\frac{y_j^2}{A_j}\sum_{j\in [q]}A_j(1+t\pi_j)^2,\]
so \eqref{eq:4rf4rf4678888b} and hence \eqref{eq:4rf4rf4678888} will follow if we find $t$ such that
\begin{equation}\label{eq:tg5tgygy4fr4}
\sum_{j\in [q]}A_j(1+t\pi_j)^2\leq q.
\end{equation}
We will choose $t$ to minimise the l.h.s. in \eqref{eq:tg5tgygy4fr4}, i.e., set
\begin{equation}\label{eq:tchoice}
t:=-\frac{\sum_{j\in [q]} A_j\pi_j }{\sum_{j\in [q]} A_j\pi_j^2},\quad \mbox{ so that } \quad \sum_{j\in [q]}A_j(1+t\pi_j)^2=\sum_{j\in [q]} A_j-\frac{\big(\sum_{j\in [q]} A_j\pi_j\big)^2}{\sum_{j\in [q]} A_j\pi_j^2}.
\end{equation}
Therefore, for this choice of $t$, \eqref{eq:tg5tgygy4fr4} becomes
\begin{equation}\label{eq:tg5tgygy4fr4b}
\sum_{j\in [q]}(A_j-1)\sum_{j\in [q]} A_j\pi_j^2\leq \bigg(\sum_{j\in [q]} A_j\pi_j\bigg)^2.
\end{equation}
Using that $A_j=\frac{C^2}{C^2-\pi_j^2\gamma_j^2}$, \eqref{eq:tg5tgygy4fr4b} is equivalent to (note the division by $C^2$ of both sides)
\begin{equation}\label{eq:tg5tgygy4fr4c}
\sum_{j\in [q]}\frac{\pi_j^2\gamma_j^2}{C^2-\pi_j^2\gamma_j^2}\sum_{j\in [q]} \frac{\pi_j^2}{C^2-\pi_j^2\gamma_j^2}\leq \bigg(\sum_{j\in [q]} \frac{C\pi_j}{C^2-\pi_j^2\gamma_j^2}\bigg)^2.
\end{equation}
We next establish \eqref{eq:tg5tgygy4fr4c}. We can upper bound the l.h.s. of \eqref{eq:tg5tgygy4fr4c} using the inequality $ab\leq \big(\frac{a+b}{2}\big)^2$, which gives that
\[\sum_{j\in [q]}\frac{\pi_j^2\gamma_j^2}{C^2-\pi_j^2\gamma_j^2}\sum_{j\in [q]}\frac{\pi_j^2}{C^2-\pi_j^2\gamma_j^2}\leq \bigg(\sum_{j\in [q]}\frac{\pi_j^2(1+\gamma_j^2)}{2(C^2-\pi_j^2\gamma_j^2)}\bigg)^2.\]
So, to prove \eqref{eq:tg5tgygy4fr4c}, it suffices to show that for each $i\in [q]$, it holds that
\[\frac{\pi_j^2(1+\gamma_j^2)}{2(C^2-\pi_j^2\gamma_j^2)}\leq  \frac{C\pi_j}{C^2-\pi_j^2\gamma_j^2}\]
which is indeed true, since $C\geq \frac{1}{2}\pi_j(1+\gamma_j^2)$ for all $i\in [q]$ by the definition of $C$. 

This proves \eqref{eq:tg5tgygy4fr4c}, which in turn establishes \eqref{eq:4rf4rf4678888b} for the choice of $t$ in \eqref{eq:tchoice}. This yields \eqref{eq:4rf4rf4678888} and hence \eqref{eq:main1main1} as well, finishing the proof of Lemma~\ref{lem:mainlemma2}.
\end{proof}

\section{Gradient analysis with blocked colors: proof of Lemma~\ref{lem:mainlemma1}}\label{sec:mainlemma1}
In this section, we prove Lemma~\ref{lem:mainlemma1}.
\begin{proof}[Proof of Lemma~\ref{lem:mainlemma1}]
For  $i\in [d]$ and $j\in [q]$, let $F^{(i)}_{c,j}(\mathbf{x})$ be the partial derivative $\frac{\partial f_c}{\partial x_{i,j}}$ viewed as a function of the ``concatenated'' vector $\xb=(\xb_1,\hdots, \xb_d)$. Note that, whenever $x_{i,j}\neq 1$, we have that
\begin{equation}\label{eq:partialderivatives}
\begin{aligned}
F^{(i)}_{c,j}(\mathbf{x})&=-\frac{f_c(\xb_1,\hdots, \xb_d)-\big(f_c(\xb_1,\hdots, \xb_d))^2}{1-x_{i,j}}\mbox{ if $j= c$},\\
F^{(i)}_{c,j}(\mathbf{x})&=\frac{f_c(\xb_1,\hdots, \xb_d)f_j(\xb_1,\hdots, \xb_d)}{1-x_{i,j}}\mbox{ if $j\neq c$}.
\end{aligned}
\end{equation}
As mentioned earlier, we will interpolate between $\pib$ and $\pib'$ by interpolating along the straight-line segment connecting  $(\pib_1,\hdots,\pib_d)$ and $(\pib_1',\hdots,\pib_d')$. In particular, for $t\in [0,1]$, let $\hat{\pi}_c(t)$ denote the $c$-th entry of the vector $\hat{\pib}(t)$ defined in the statement of the lemma. Then, we have that 
\begin{equation}\label{eq:expressmarginals}
\hat{\pi}_c(t)=f_c(\zb(t)), \mbox{ where   $\zb(t)$ is the vector $\big(t\pib_1+(1-t)\pib_1',\hdots, t\pib_d+(1-t)\pib_d'\big)$}.
\end{equation} 
We will use $z_{i,j}(t)$ to denote the $j$-th entry of the $i$-th vector in $\zb(t)$, i.e., $z_{i,j}(t)=t\pi_{i,j}+(1-t)\pi_{i,j}'$. 

Let $D$ be the set of indices $i$ such that $v_i$ is not frozen under $\eta$ and $\eta'$ (cf. Observation~\ref{obs:sameblocked}). Observe that, for all $i\notin D$ and $c,j\in [q]$, we have that $z_{i,j}(t)=\pi_{i,j}=\pi_{i,j}'$ for $t\in[0,1]$. Moreover, for $i\in D$ and $j\in [q]$ we have that $\pi_{i,j},\pi_{i,j}'\leq 1/(q-d)$ (since the child $v_i$ has at least $q-d$ available colors in the subtree $T_i$) and  hence
\begin{equation}\label{eq:zijtineq}
0\leq z_{i,j}(t)\leq 1/(q-d).
\end{equation} 
Since $z_{i,j}(t)\neq 1$ for $i\in D$ and $j\in [q]$, it follows that
\[\frac{d\hat{\pi}_c}{dt}=\sum_{i\in D}\sum^q_{j=1}F^{(i)}_{c,j}(\zb(t))(\pi_{i,j}-\pi_{i,j}').\]
Using \eqref{eq:tree1recursion}, we  therefore have that 
\begin{align*}
(\pi_c-\pi_c')^2&=\big(\hat{\pi}_c(1)-\hat{\pi}_c(0)\big)^2=\Big(\int^1_{0} \frac{d\hat{\pi}_c}{dt} dt\Big)^2=\Big(\int^{1}_{0}\sum_{i\in D}\sum^q_{j=1}F^{(i)}_{c,j}(\zb(t))(\pi_{i,j}-\pi_{i,j}')dt\Big)^2\\
&\leq \int^{1}_{0}\bigg(\sum_{i\in D}\sum^q_{j=1}F^{(i)}_{c,j}(\zb(t))(\pi_{i,j}-\pi_{i,j}')\bigg)^{2}dt,
\end{align*}
where the last inequality follows by applying the Cauchy-Schwarz inequality for integrals.
By summing over all colors $c\in [q]$, we obtain
\begin{equation}\label{eq:milestone1}
\big\| \pib-\pib'\big\|_2^2\leq \int^{1}_{0}\sum^{q}_{c=1}\bigg(\sum_{i\in D}\sum^q_{j=1}F^{(i)}_{c,j}(\zb(t))(\pi_{i,j}-\pi_{i,j}')\bigg)^{2}dt.
\end{equation}
To simplify the r.h.s. of \eqref{eq:milestone1}, we first note that, by \eqref{eq:partialderivatives} and \eqref{eq:expressmarginals}, we have   
\begin{equation}\label{eq:Acjtdef}
F^{(i)}_{c,j}(\zb(t))=\frac{A_{c,j}(t)}{1-z_{i,j}(t)} \mbox{ where } A_{c,j}:=\left\{\begin{array}{ll} \big(\hat{\pi}_c(t))^2-\hat{\pi}_c(t),& \mbox{ if $j= c$},\\[0.1cm] \hat{\pi}_c(t)\hat{\pi}_j(t), & \mbox{ if $j\neq c$}\end{array}\right.
\end{equation}
Moreover, for $j\in [q]$, set 
\begin{equation}\label{eq:4by6y5yyby}
u_j(t)=\frac{1}{|D|\gamma_j}\sum_{i\in D}\frac{\pi_{i,j}-\pi_{i,j}'}{1-z_{i,j}(t)} \mbox{ if $\gamma_j>0$, else set $u_j(t)=0$}.
\end{equation}
Note that if color $j$ is blocked for the child $v_i$ we have that $\pi_{i,j}-\pi_{i,j}'=0$, so  using the power mean inequality we have that
\begin{equation}\label{eq:powermean}
\gamma_j(u_j(t))^2\leq  \frac{1}{|D|}\sum_{i\in D}\Big(\frac{\pi_{i,j}-\pi_{i,j}'}{1-z_{i,j}}\Big)^2.
\end{equation}
Then, for $c\in [q]$, we have that 
\begin{equation}\label{eq:gb4by7657h}
\sum_{i\in D}\sum^q_{j=1}F^{(i)}_{c,j}(\zb(t))(\pi_{i,j}-\pi_{i,j}')=\sum^q_{j=1}A_{c,j}(t)\sum_{i\in D}\frac{\pi_{i,j}-\pi_{i,j}'}{1-z_{i,j}(t)}=|D|\sum^q_{j=1}A_{c,j}(t)\gamma_j u_j(t),
\end{equation}
where the last equality follows from \eqref{eq:4by6y5yyby} and observing that if $\gamma_j=0$ then $\pi_{i,j}-\pi_{i,j}'=0$ for all $i\in D$. Note that the $(c,q)$-entry of $\Mb_{\hat{\pib}(t),\sqrt{\gammab}}$ is exactly $-A_{c,j}(t)\sqrt{\gamma_j}$ (cf. \eqref{eq:Acjtdef} and Definition~\ref{def:gradientmatrix})  and hence, using \eqref{eq:gb4by7657h}, we can write the integrand in the r.h.s. of \eqref{eq:milestone1} as
\begin{equation}\label{eq:4tg45g5g123}
\sum^{q}_{c=1}\bigg(\sum_{i\in D}\sum^q_{j=1}F^{(i)}_{c,j}(\zb(t))(\pi_{i,j}-\pi_{i,j}')\bigg)^{2}=|D|^2\, \left\|\Mb_{\hat{\pib}(t),\sqrt{\gammab}}\ub(t)\right\|_2^2,
\end{equation}
where, for $t\in[0,1]$, $\ub(t)$ is the $q$-dimensional vector with entries $\{\sqrt{\gamma_j}\, u_j(t)\}_{j\in [q]}$. Let 
\[W:=\max_{t\in [0,1]}\left\|\Mb_{\hat{\pib}(t),\sqrt{\gammab}}\right \|_2, \mbox{ so that } K=\frac{W}{1-\frac{1}{q-d}}.\]
Then, for $t\in [0,1]$, we have that 
\begin{equation}\label{eq:vttg55ttta}
\begin{aligned}
\left\|\Mb_{\hat{\pib}(t),\sqrt{\gammab}} \ub(t)\right\|_2^2&\leq W^2 \left\|\ub(t)\right \|_2^2= W^2 \sum_{j\in [q]} \gamma_j (u_j(t))^2\leq \frac{W^2}{|D|}\sum_{j\in [q]}\sum_{i\in D} \left\|\frac{\pi_{i,j}-\pi_{i,j}'}{1-z_{i,j}(t)}\right\|_2^2\\
&\leq  \frac{K^2}{|D|}\sum_{j\in [q]}\sum_{i\in D} \left\|\pi_{i,j}-\pi_{i,j}'\right\|_2^2=  \frac{K^2}{|D|}\sum_{i\in [d]} \left\|\pib_i-\pib_i'\right\|_2^2,
\end{aligned}
\end{equation}
where the first inequality is by definition of the norm, the second inequality follows from  \eqref{eq:powermean}, the third inequality follows from $0\leq z_{i,j}(t)\leq 1/(q-d)$, and the last equality follows from the fact that for $i\notin D$ we have that $\pib_i=\pib_i'$. Combining \eqref{eq:milestone1}, \eqref{eq:4tg45g5g123} and \eqref{eq:vttg55ttta}, we obtain that
\[\big\| \pib-\pib'\big\|_2^2\leq |D|K^2\sum_{i\in [d]} \left\|\pib_i-\pib_i'\right\|_2^2.\]
This finishes the proof of Lemma~\ref{lem:mainlemma1}.
\end{proof}

\section{Bounds on the marginals: proof of Lemma~\ref{lem:mainlemma3}}\label{sec:mainlemma3}
In this section, we prove Lemma~\ref{lem:mainlemma3}. We begin with the following lemma.
\begin{lemma}\label{lemma:LowerBound4Marg}
Let $q,d,h$ be positive integers so that $q\geq d+1$ and $h\geq 1$. Let $T=\mathbb{T}_{d,h, \rho}$ be the $d$-ary tree with height $h$ rooted at $\rho$, $\Lambda$ be a subset of the vertices of $T$ such that $\rho\notin \Lambda$, and $\eta:\Lambda\rightarrow [q]$ be an extendible assignment of $T$. Then, for all colors $k\in [q]$ that are available for $\rho$ under $\eta$, it holds that
\[\mu_T(\sigma_\rho=k\mid \sigma_\Lambda=\eta)\geq \frac{\big(1-\frac{1}{q-d}\big)^{d}}{d+(q-d)\big(1-\frac{1}{q-d}\big)^{d}}.\]
\end{lemma}
\begin{proof}
Let $Q\subseteq [q]$ be the set of all colors that are available for $\rho$ under $\eta$ and let $k\in Q$. Let  $v_1,\hdots, v_{d}$ be the children of $\rho$ in $T$  and let $D=\{i\in [d]\mid v_i\notin \Lambda\}$ be the indices of the children of $\rho$ that do not belong to $\Lambda$. 

For $i\in [d]$, let $T_i=(V_i,E_i)$ be the subtree of $T$ rooted at $v_i$ which consists of all descendants of $v_i$ in $T$ (together with $v_i$ itself). Further, for a color $j\in [q]$, let 
\[x_{i,j}=\mu_{T_i}\big(\sigma_{v_i}=j\mid \sigma_{\Lambda\cap V_i}=\eta_{\Lambda\cap V_i}\big), \]
i.e., $x_{i,j}$ is the marginal probability that $v_i$ takes the color $j$ at $v_i$ in $\mu_{T_i}$ with boundary condition $\eta_{\Lambda\cap V_i}$. Note that 
\begin{equation}\label{eq:Constraint4Marginals}
0\leq x_{i,j}\leq \frac{1}{q-d} \mbox{ for all $i \in D$ and $j\in [q]$}, \qquad  \sum_{j\in [q]}x_{i,j}=1 \mbox{ for all $j\in [q]$}.
\end{equation}
Using the tree recursion \eqref{def:fc} and ignoring summands that are 0 or factors that are equal to 1, the marginal $\mu_T(\sigma_\rho=k\mid \sigma_\Lambda=\eta)$  is expressed in terms of $x_{i,j}$ as follows:
\begin{equation}\label{eq:v3tvtvtedew}
\mu_T(\sigma_\rho=k\mid \sigma_\Lambda=\eta)=\frac{\prod_{i\in D} (1-x_{i,k})}{\sum_{j\in Q}\prod_{i\in D} (1-x_{i,j})}.
\end{equation}
We prove the lemma by deriving an appropriate lower bound on the quantity at the r.h.s. of \eqref{eq:v3tvtvtedew}
subject to the  constraint in \eqref{eq:Constraint4Marginals}. For the numerator in \eqref{eq:v3tvtvtedew}, we have that
\begin{equation}\label{eq:4tv4t4tgt4}
\prod_{i\in D} (1-x_{i,k})\geq  \left(1-\frac{1}{q-d}\right)^{|D|}.
\end{equation}
For the denominator we are going to show the following:  
\begin{equation}\label{eq:UpBDenominator4v3tvtvtedew}
\sum_{j\in Q}\prod_{i\in D} (1-x_{i,j}) \leq d+(q-d)\left(1-\frac{1}{q-d}\right)^{|D|}.
\end{equation}
Before showing that \eqref{eq:UpBDenominator4v3tvtvtedew} is indeed true,  note that the lemma follows by plugging  \eqref{eq:4tv4t4tgt4}, \eqref{eq:UpBDenominator4v3tvtvtedew} into \eqref{eq:v3tvtvtedew}, yielding
$$
\mu_T(\sigma_\rho=k\mid \sigma_\Lambda=\eta) \geq \frac{\left(1-\frac{1}{q-d}\right)^{|D|}}{d+(q-d)\left(1-\frac{1}{q-d}\right)^{|D|}} \geq 
\frac{\left(1-\frac{1}{q-d}\right)^{d}}{d+(q-d)\left(1-\frac{1}{q-d}\right)^{d}},
$$
where the last inequality follows by noting that the ratio in the middle is decreasing in $|D|$ and $|D|\leq d$.

We now proceed with the proof of \eqref{eq:UpBDenominator4v3tvtvtedew}. 
First, we have the simple bound
\begin{equation}\label{eq:rtv5tvt121e}
\sum_{j\in Q}\prod_{i\in D} (1-x_{i,j})\leq \sum_{j\in [q]} \prod_{i\in D} (1-x_{i,j}).
\end{equation}
For $j\in [q]$, let $x_j=\frac{1}{|D|}\sum_{i\in D}x_{i,j}$ and note that $(x_1,\hdots,x_q)$ is a probability vector whose entries are in $[0,1/(q-d)]$. By the AM-GM inequality, we can bound the r.h.s. of \eqref{eq:rtv5tvt121e} by
\begin{equation}\label{eq:t554g5g56gfe}
\sum_{j\in [q]} \prod_{i\in D} (1-x_{i,j})\leq \sum_{j\in [q]} (1-x_{j})^{|D|}.
\end{equation}
It remains to observe that the function $f(\mathbf{z})=\sum_{j\in [q]}(1-z_j)^{|D|}$ is convex over the space of probability vectors $\mathbf{z}=(z_1,\hdots,z_q)$ whose entries are in $[0,1/(q-d)]$, and hence $f$ attains its maximum at the extreme points of the space, which are given by (the permutations of) the probability vector whose first $d$ entries are  equal to zero and the rest are equal to $1/(q-d)$. It follows that
\begin{equation}\label{eq:344f5f45111}
\sum_{j\in [q]} (1-x_{j})^{|D|}\leq d+(q-d)\Big(1-\frac{1}{q-d}\Big)^{|D|}.
\end{equation}
Combining \eqref{eq:rtv5tvt121e}, \eqref{eq:t554g5g56gfe} and \eqref{eq:344f5f45111} yields \eqref{eq:UpBDenominator4v3tvtvtedew}, thus concluding the proof of Lemma~\ref{lemma:LowerBound4Marg}.
\end{proof}

We are now ready to prove Lemma~\ref{lem:mainlemma3}.
\begin{proof}[Proof of Lemma~\ref{lem:mainlemma3}]
For convenience, let $r=1.59$, so that $q/d\geq r$. We will use that $r$ satisfies
\begin{equation}\label{eq:Cdefdef}
C:=\frac{1}{r}\exp\Big(\frac{1}{r}\Big)\exp\Big(-\frac{1}{r - 1 +\exp\big(\frac{1}{r-1}\big)}\Big)<1.
\end{equation}
We will show the result with the constant $K'=(1+C)/2$. For the rest of this proof, we will focus on the case $q\in [1.59d+\beta,2.01d]$, for some large constant $\beta>0$ (when $q>2.01d$ the desired bound follows rather crudely, see Footnote~\ref{fn:q2d} below for details).

Recall that  $v_1,\hdots, v_{d}$ are the children of $\rho$ in $T$  and  $D$ is the set of (indices of the) non-frozen children of the root $\rho$.  Let $Q\subseteq [q]$ be the set of all colors that are available for $\rho$ under $\eta$; since at most $d-|D|$ colors can be blocked for $\rho$, we have that
\begin{equation}\label{eq:QqdDaq}
|Q|\geq q-(d-|D|).
\end{equation} 
For $i\in [d]$, let $T_i=(V_i,E_i)$ be the subtree of $T$ rooted at $v_i$ which consists of all descendants of $v_i$ in $T$ (together with $v_i$ itself). Further, for a color $j\in [q]$, recall that
\begin{equation}
\begin{aligned}
\pi_{i,j}&=\mu_{T_i}\big(\sigma_{v_i}=j\mid \sigma_{\Lambda\cap V_i}=\eta_{\Lambda\cap V_i}\big),\\
\pi_{i,j}'&=\mu_{T_i}\big(\sigma_{v_i}=j\mid \sigma_{\Lambda\cap V_i}=\eta_{\Lambda\cap V_i}'\big), 
\end{aligned}
\end{equation}
i.e., $\pi_{i,j}$ is the marginal probability that $v_i$ takes the color $j$ at $v_i$ in $\mu_{T_i}$ with boundary condition $\eta_{\Lambda\cap V_i}$. For a non-frozen child $v_i$ (i.e., $i\in D$), note that, if color $j$ is available for $v_i$ (in $T_i$), then we have from Lemma~\ref{lemma:LowerBound4Marg} the bounds
\begin{equation}\label{eq:rf4ff4ftghhty}
L\leq \pi_{i,j},\pi_{i,j}, \mbox{ where } L=\frac{\big(1-\frac{1}{q-d}\big)^{d}}{d+(q-d)\big(1-\frac{1}{q-d}\big)^{d}}.
\end{equation}
Another useful bound to observe for later is that 
\[d L<1/3 \mbox{ for all $d\geq 2$.}\]

Consider arbitrary $k\in Q$. For $t\in[0,1]$, let $\zb(t)$ be the vector $\big(t\pib_1+(1-t)\pib_1',\hdots, t\pib_d+(1-t)\pib_d'\big)$. Using the tree recursion \eqref{def:fc} and ignoring summands that are 0 or factors that are equal to 1, we obtain
\begin{equation}\label{eq:v3tvtvtedew12}
\hat{\pi}_k(t)=\frac{\prod_{i\in D} (1-z_{i,k}(t))}{\sum_{j\in Q}\prod_{i\in D} (1-z_{i,j}(t))}.
\end{equation}
Recall,  our goal is to show that $\frac{1}{2}\hat{\pi}_k(t)(1+\gamma_k)<K'/|D|$ for all $t\in [0,1]$, where $\gamma_k\in [0,1]$ is the fraction of non-frozen children that have color $k$ available. \footnote{\label{fn:q2d}For $q>2.01d$, we have from \eqref{eq:v3tvtvtedew12} and \eqref{eq:QqdDaq} that $\hat{\pi}_k(t)\leq \frac{1}{|Q|-|D|}\leq \frac{1}{q-d}<\frac{1}{1.01d}\leq K'/|D|$, yielding the desired inequality.}Note that,  if color $j$ is available for the child $v_i$, \eqref{eq:rf4ff4ftghhty} gives that
\begin{equation*}
L\leq z_{i,j}(t)\mbox{ for } t\in[0,1],
\end{equation*}
so, using the fact that the color $k$ is available for $|D|\gamma_k$ non-frozen children, we obtain that the numerator of \eqref{eq:v3tvtvtedew12} is bounded by
\begin{equation}\label{eq:tg4tgvrf3fgty66a}
\prod_{i\in D} \big(1-z_{i,k}(t)\big)\leq (1-L)^{|D|\gamma_k}\leq \exp(-L|D|\gamma_k),
\end{equation}
whereas the denominator, using the AM-GM inequality analogously to \cite[Lemma 2.1 \& Corollary 2.2]{DF}, by
\begin{equation}\label{eq:tg4tgvrf3fgty66b}
\begin{aligned}
\sum_{j\in Q}\prod_{i\in D} \big(1-z_{i,j}(t)\big)&=\sum_{j\in [q]}\prod_{i\in D} \big(1-z_{i,j}(t)\big)-\sum_{j\in [q]\backslash Q}\prod_{i\in D} \big(1-z_{i,j}(t)\big)\\
&\geq  \big(q \exp(-|D|/q)-\tau) -(d-|D|),
\end{aligned}
\end{equation}
where $\tau>0$ is an absolute constant (independent of $q,d,\beta$). From \eqref{eq:v3tvtvtedew12}, \eqref{eq:tg4tgvrf3fgty66a}, and  \eqref{eq:tg4tgvrf3fgty66b}, it follows that $\hat{\pi}_k(t)\leq \frac{\exp(-L|D|\gamma_k)}{q \exp(-|D|/q) -(d-|D|)-\tau}$. Therefore, the lemma will follow by showing that 
\begin{equation}\label{eq:rv4g4ygh112}
\frac{|D|\exp(-|D|L\gamma_k)}{q \exp(-|D|/q) -(d-|D|)-\tau}(1+\gamma_k)<2K'.
\end{equation}
Note that the function $h(x)=(1+x) \exp(-d L x)$ is increasing when $x\in [0,1]$, since
\[h'(x)=\exp(-d Lx)\big(1-d L(1+x)\big)\geq \exp(-d Lx)(1-2dL)>0.\]
Therefore, to prove \eqref{eq:rv4g4ygh112}, it suffices to show that
\begin{equation}\label{eq:rv4g4ygh112av}
\frac{|D|\exp(-|D|L)}{q \exp(-|D|/q) -(d-|D|)-\tau}<K', \mbox{ or equivalently that } f(|D|)>0
\end{equation}
where $f(x):=K'\big(q \exp(-x/q) -d+x-\tau\big)-x\exp(-Lx)$ for $x\in [0,d]$. We claim that $f(x)$ is decreasing in $x$. We have
\[f'(x)=K'-K'\exp(-x/q)-\exp(-Lx)(1-Lx)\]
which is maximised for $x=d$. In particular,
\begin{align*}
f'(x)&\leq f'(d)=K'-K'\exp(-d/q)-\exp(-dL)(1-dL)\\
&\leq K'-K'\exp(-1/r)-\exp(-1/3)(1-1/3)\leq 0,
\end{align*}
where the second to last inequality follows from the fact that $dL<1/3$ and the last inequality using that $K'<1$. For $|D|=d$, \eqref{eq:rv4g4ygh112av} becomes
\begin{equation}\label{eq:rv4g4ygh112av124}
\frac{d\exp(-dL)}{q \exp(-d/q)-\tau}<K'.
\end{equation}

Now, we have that
\[d L\geq\frac{1}{r - 1 +\exp\big(\frac{d}{(r-1)d-1}\big)}.\]
Therefore, by choosing $\beta$ large enough and using that $q\in [1.59d+\beta,2.01d]$, we can ensure that
\[\frac{d\exp(-dL)}{q \exp(-d/q)-\tau}<\frac{1+C}{2}=K',\]
where $C$ is the constant in \eqref{eq:Cdefdef}. This proves \eqref{eq:rv4g4ygh112av124} and therefore concludes the proof of Lemma~\ref{lem:mainlemma3}.
\end{proof}

\section{Proof of Theorem~\ref{thm:main}}\label{sec:proofofmaintheorem}
Finally, utilizing Theorem~\ref{thm:main2}, we give the proof of Theorem~\ref{thm:main}. 
\begin{proof}[Proof of Theorem~\ref{thm:main}]
From Theorem~\ref{thm:main2}, we know that there exist constants $\beta>0$ and $U\in(0,1)$ such that for all $q\geq 1.59d+\beta$ the conclusion of Theorem~\ref{thm:main2} applies.   Note that Theorem~\ref{thm:main2} applies to the $d$-ary tree rather than the $(d+1)$-regular tree but these trees  differ only at the degree of the root.  To account for it, we will assume that $q\geq 1.59(d+1)+\beta$, i.e., prove Theorem~\ref{thm:main} with constant $\beta'=\beta+1.59$. Consider the function $\zeta$ given by $\zeta(\ell)=2U^{\ell-2}$ for $\ell\geq 0$ and note that $\zeta$ is exponentially decaying. We will show that the $q$-coloring model  has strong spatial mixing on the $(d+1)$-regular tree with decay rate $\zeta$. 

We first show by induction on $h$ that, for the tree $T=\hat{\mathbb{T}}_{d+1,h, \rho}$ (that is, the $(d+1)$-ary tree with height $h$ rooted at $\rho$),  for any subset $\Lambda$ of vertices of $T$ and arbitrary extendible assignments $\eta,\eta':\Lambda\rightarrow [q]$ of $T$,  it holds that 
\begin{equation}\label{eq:trb35b65}
\big\| \pib_{T,\rho, \eta}-\pib_{T,\rho, \eta'}\big\|_2^2\leq \zeta(\mathrm{dist}(\rho,\Delta)),
\end{equation}
where $\Delta\subseteq \Lambda$ is the set of vertices where $\eta$ and $\eta'$ disagree. The base cases $h=0,1,2$ are trivial so assume $h\geq 3$ in what follows. Let $\ell=\mathrm{dist}(\rho,\Delta)$. Once again,  \eqref{eq:trb35b65} is trivial when $\ell\leq 2$, so assume $\ell\geq 3$ in what follows.  Let $v_1,\hdots, v_{d+1}$ be the children of $\rho$ and, for $i\in [d+1]$, let $T_i=(V_i,E_i)$ be the subtree of $T$ rooted at $v_i$ which consists of all descendants of $v_i$ in $T$. Further, let  $\pib_i=\pib_{T_i,v_i, \eta(\Lambda\cap V_i)}$, $\pib_i'=\pib_{T_i, v_i, \eta'(\Lambda\cap V_i)}$. Then, by Theorem~\ref{thm:main2} and since $q\geq 1.59(d+1)+\beta$, we have that
\begin{equation}\label{eq:tv3t5t}
\big\| \pib_{T,\rho, \eta}-\pib_{T,\rho, \eta'}\big\|_2^2\leq U\max_{i\in[d+1]}\big\| \pib_{i}-\pib_{i'}\big\|_2^2.
\end{equation} 
For $i\in [d+1]$, since $T_i$ is isomorphic to $\hat{\mathbb{T}}_{d+1,h-1,\rho}$ we have by the induction hypothesis that 
\[\big\| \pib_{i}-\pib_{i'}\big\|_2^2\leq \zeta(\ell-1).\]
Combining this with \eqref{eq:tv3t5t} and the fact that $\zeta(\ell)=U \zeta(\ell-1)$ yields \eqref{eq:trb35b65}, completing the induction and therefore that strong spatial mixing holds on $T$ with decay rate $\zeta$.

Now, let $T=(V,E)$ be a finite subtree of the $(d+1)$-regular tree, $v$ be an arbitrary vertex of $T$, $\Lambda$ be a subset  of vertices of $T$ and $\eta,\eta':\Lambda\rightarrow [q]$ be arbitrary extendible assignments of $T$. Then, we can view $T$ as a subgraph of $T_v=\hat{\mathbb{T}}_{d+1,h, v}$ for some appropriate height $h$. It also holds that (see, for example, \cite[Lemma 25]{Efthymiou})
\[\| \pib_{T,v, \eta}-\pib_{T,v, \eta'}\big\|_2=\| \pib_{T_v,v, \eta}-\pib_{T_v,v, \eta'}\big\|_2.\]
Therefore,  from \eqref{eq:trb35b65} (applied to the tree $T_v$) we obtain that
\[\big\| \pib_{T,v, \eta}-\pib_{T,v, \eta'}\big\|_2^2\leq \zeta(\mathrm{dist}(v,\Delta)),\]
where $\Delta\subseteq \Lambda$ is the set of vertices where $\eta$ and $\eta'$ disagree.  

This completes the proof of Theorem~\ref{thm:main}.
\end{proof}

\bibliographystyle{plain}
\bibliography{bibliography}

\end{document}